\documentclass[12 pt]{article}%
\usepackage{amsmath,amssymb,amsfonts,latexsym}
\usepackage{amsfonts}
\usepackage{amssymb}
\usepackage{graphicx}
\usepackage{amsmath}%
\usepackage[colorlinks=true]{hyperref}
\providecommand{\U}[1]{\protect\rule{.1in}{.1in}}
\newtheorem{theorem}{Theorem}
\newtheorem{proof}{Proof}
\newtheorem{proposition}{Proposition}
\providecommand{\U}[1]{\protect\rule{.1in}{.1in}}

\begin{document}

\title{On Integrals, Hamiltonian and Metriplectic Formulations of 3D Polynomial Systems}
\author{O\u{g}ul Esen\thanks{E-mail: oesen@gtu.edu.tr}\\Department of Mathematics, Gebze Technical University \\Gebze-Kocaeli 41400, Turkey. \\ \\
\and Anindya Ghose Choudhury\thanks{Email: aghosechoudhury@gmail.com}\\Department of Physics, Surendranath College,\\24/2 Mahatma Gandhi Road, Calcutta-700009, India.\\ \\
\and Partha Guha\thanks{E-mail: partha@bose.res.in}\\S.N. Bose National Centre for Basic Sciences \\JD Block, Sector III, Salt Lake \\Kolkata - 700098, India \\
}
\maketitle

\begin{abstract}
We apply the Darboux integrability method to determine first integrals and
Hamiltonian formulations of three dimensional polynomial systems; namely the
reduced three-wave interaction problem, the Rabinovich system, the
Hindmarsh-Rose model, and the oregonator model. Additionally, we investigate
their Hamiltonian, Nambu-Poisson and metriplectic characters.


\end{abstract}

\newpage

\tableofcontents

\section{Introduction}

The problem of solving ordinary nonlinear differential equations is a
challenging area in nonlinear dynamics. For a two dimensional system the
existence of a first integral completely determines its phase portrait. In
these cases chaos cannot arise because of the Poincar\'{e}-Bendixson theorem
\cite{HS} which says that any limit of a $2$D system of differential equation
is either a fixed point or a cycle. In three dimension this is no longer true.
In the case of non planar systems the problem of determining first integrals
is a non trivial task in general, and various methods have been introduced for
studying the existence of first integrals. However, except for some special
cases \cite{Hi} there are few known satisfactory methods to solve it in general.

Non planar systems are often non-Hamiltonian in character and describe the
time evolution of physical processes which are usually dissipative in nature.
In general, a Pfaff differential form in $n$-dimensions
\[
F_{1}(x_{1},\cdots,x_{n})dx_{1}+\cdots+F_{n}(x_{1},\cdots,x_{n})dx_{n}%
\]
is not exact and therefore integrating factor may not be exist. Earlier a
direct method \cite{GRZ} has been used to search for a first integral of three
dimensional dynamical systems. This method consists in proposing an ansatz for
the invariant which is a polynomial of a given degree in one coordinates of
the phase space of the system. So reader can see immediately that this is a
tedious method applied to a very special class of systems. In fact,
Grammaticos \textit{et al} \cite{Gr} proposed another method, based on the
Frobenius integrability theorem, for finding integrals for three-dimensional
ordinary differential equations. None of these methods are extremely
successful. In a similar programme. Dorizzi \textit{et al} \cite{DGH}
investigated a three-dimensional Hamiltonian systems with quartic potentials
that are even in $x$, $y$, and $z$. They applied reduction method to obtain
two new integrable systems and their constants of motion.

One might ask why do we need first integrals. An integral defines an invariant
manifold for the flow which can be used to eliminate one degree of freedom.
When the system admits an integral of motion, the analysis of its dynamical
behaviour, especially in $t\rightarrow\infty$ limit, is greatly simplified. As
elucidated by Giacomini and Neukrich \cite{GN,GN2}, the first integrals can be
used in the non integrable regimes to build generalized Lyapunov functions and
obtain bounds on the chaotic attractors of three-dimensional vector fields
and prove the absence of homoclinic orbits. Therefore computing the first
integral is an important problem but unfortunately the problem of finding a
first integral is mathematically the same problem as solving the original
system. Indeed exact first integrals are known only in special cases.

In this paper, we are interested in the integrability of the polynomial
differential systems of $3$ dimensions. A polynomial system is said to be
Darboux integrable if it possesses a first integral or an integrating factor
given by Darboux polynomial \cite{Da}. In particular, Darboux showed (see for
example \cite{CG}) that a polynomial system of degree $n$, with at least
$n(n+1)/2+1$ invariant algebraic curves, has a first integral which can be
expressed by means of these algebraic curves. Note that, the knowledge of
algebraic curves can be used to study the topological properties of the system.

The goal of this paper is to obtain the first integrals of some polynomial
three dimensional ODE systems, namely the reduced three-wave interaction
problem, Rabinovich system, Hindmarsh-Rose model and Oregonator model, using
Darboux polynomials. After deriving the first integrals, we shall further
investigate the possible Hamiltonian formulations, bi-Hamiltonian
representations or/and metriplectic realizations of these systems. We shall
derive Poisson tensors, metric tensors for each system explicitly.

In order to achieve these goals, the paper is divided into two main sections.
The following section is reserved for the theoretical background on the
notions of integrability, Hamiltonian, Nambu Poisson and metriplectic
formulations in three dimensional models. The theorem (\ref{1}) in the first
subsection has the prominent role while determining the first integrals using
the Darboux polynomials. After finding an integral of a system $\dot
{\mathbf{x}}=X$, one starts to wonder whether or not that the system is
Hamiltonian. In three dimensions, a Hamiltonian system is bi-Hamiltonian and
Nambu-Poisson if it is possible to find a Jacobi's last multiplier $M$ which
makes $MX$ divergence free (c.f. see theorem (\ref{3})). A dissipative system
can not is not Hamiltonian, but it can be written as a metriplectic
formulation which is a linear combination of a Poisson and a gradient systems.
The third section is for application of the technics presented in the section
$2$ to the particular models. For several subcases of the reduced three-wave
interaction problem, for the Rabinovich system, and for the subcases of the
Hindmarsh-Rose model, the first integrals will be constructed. A
bi-Hamiltonian/Nambu metriplectic formulation of these systems will be
exhibited. First integrals of the Oregonator model is established and the
model written as a Hamiltonian system.
\section{Some Theory on 3D Polynomial Systems}
\subsection{Darboux' Polynomials} \label{dp}

A three dimensional polynomial ODE system is given by the set of equations
\begin{equation}
\dot{x}=P(\mathbf{x}),\qquad\dot{y}=Q(\mathbf{x}),\qquad\dot{z}=R(\mathbf{x}),
\label{e1}%
\end{equation}
where $P,Q,R$ are real valued polynomials with real coefficients. Here, the
boldface $\mathbf{x}$ stands for the three tuple $\left(  x,y,z\right)  $. The
degree $m$ of a system is the maximum of degrees of the coefficient
polynomials. The system (\ref{e1}) defines a polynomial vector field
$X=X\left(  \mathbf{x}\right)  $ by the identity $\mathbf{\dot{x}}=X\left(
\mathbf{x}\right)  $.

A function $I=I(t,x,y,z)$ is the first integral if it is constant on any
integral curve of the system, that is if the total derivative of $I$ with
respect to $t$ vanishes on the solution curves. A second integral $g$ of a
system $\mathbf{\dot{x}}=X\left(  \mathbf{x}\right)  $ is a function
satisfying%
\begin{equation}
X(g)=\lambda g %
\end{equation}
for some function $\lambda$ called the cofactor. Polynomial second integrals
for the polynomial vector fields are called the Darboux polynomials. The
Darboux polynomials simplify the determination of possible first integrals
\cite{Da}. For example, if there exist two relatively prime Darboux
polynomials, say $P_{1}$ and $P_{2}$, having a common cofactor then their
fraction $P_{1}/P_{2}$ is a rational first integral of the polynomial vector
field $X$. The inverse of this statement is also true that is, if we have a
rational first integral $P_{1}/P_{2}$ of a vector field $X$, then $P_{1}$ and
$P_{2}$ are Darboux polynomials for $X$.

For the case of planar polynomial vector fields, there are more strong tools
for the determination of the first integrals. In \cite{PrSi83,Si92}, a
semi-algorithm, called Prelle-Singer method, is presented for the
determinations of elementary first integrals for planar systems. If we have a
certain number of relatively prime irreducible Darboux polynomials, not
necessarily having a common cofactor, it is possible to write first integrals
using the Darboux polynomials \cite{Da,DuLlAr06,Ju79,Si92}. Unfortunately,
this algorithm cannot be applicable for non-planar systems. However, Darboux
polynomials are still useful though at times the use of a specific ansatz or a
polynomial in one variable (of particular degree) with coefficients depending
on the remaining variables remains the only option. One may at times use a
variant of the Prelle-Singer/Darboux method to derive what are called
quasi-rational first integrals \cite{Ma94}. Now, we state the following
observation which enables one to arrive a time dependent first integral of a
given system when it possesses autonomous Darboux polynomials.

\begin{theorem}
\label{1} If $g_{\alpha}$'s are Darboux Polynomials for an autonomous system
$\mathbf{\dot{x}}=X$ and there exist constants $n_{\alpha}$'s, not all zero,
satisfying the equality
\begin{equation}
\sum_{\alpha=1}^{k}n_{\alpha}\lambda_{\alpha}=r, \label{clue}%
\end{equation}
for some real number $r\in\mathbb{R}$, then the function
\begin{equation}
I=e^{-rt}\prod_{\alpha=1}^{k}g_{\alpha}^{n_{\alpha}} \label{Int}%
\end{equation}
is a time dependet first integral of the system $\mathbf{\dot{x}}=X$.
\end{theorem}

\begin{proof}
To prove this assertion, we compute the total derivative of the function $I$
given in (\ref{Int}) as follows
\begin{align*}
\tilde{X}\left(  I\right)   &  =\frac{\partial}{\partial t}\left(
e^{-rt}\prod_{\alpha=1}^{k}g_{\alpha}^{n_{\alpha}}\right)  +e^{-rt}X\left(
\prod_{\alpha=1}^{k}g_{\alpha}^{n_{\alpha}}\right) \\
&  =-re^{-rt}\prod_{\alpha=1}^{k}g_{\alpha}^{n_{\alpha}}+e^{-rt}\left(
\prod_{\alpha=1}^{k}g_{\alpha}^{n_{\alpha}-1}\right)  \left(  \sum_{\beta
=1}^{k}(n_{\beta}g_{1}...X(g_{\beta})...g_{k})\right) \\
&  =-re^{-rt}\prod_{\alpha}g_{\alpha}^{n_{\alpha}}+e^{-rt}\left(
\prod_{\alpha=1}^{k}g_{\alpha}^{n_{\alpha}}\right)  \left(  \sum_{\beta=1}%
^{k}n_{\beta}\lambda_{\beta} \right) \\
&  =-re^{-rt}\prod_{\alpha}g_{\alpha}^{n_{\alpha}}+re^{-rt}\prod_{\alpha
=1}^{k}g_{\alpha}^{n_{\alpha}}=0
\end{align*}
where in the first line we assumed that $g_{\alpha}$ is not explicitly time
dependent , we applied product rule in the second line, in the third line we
used the fact that $g_{\alpha}$'s are Darboux' polynomials by satisfying the
equalities (\ref{dp}), and finally, in the last line, we applied the equality
(\ref{clue}).
\end{proof}

To the best of our knowledge, in the literature, the case where $\sum_{\alpha
}n_{\alpha}\lambda_{\alpha}\neq r $ is still open.

\subsection{Poisson Systems in $3D$}

Poisson bracket on an $n$-dimensional space is a binary operation
$\{\bullet,\bullet\}$ on the space of real-valued smooth functions satisfying
the Leibnitz and the Jacobi identities \cite{LaPi12,LiMa12,OLV,wei83}. We
define a Poisson bracket of two functions $F$ and $H$ by
\begin{equation}
\left\{  F,H\right\}  =\nabla F\cdot N\nabla H, \label{PB}%
\end{equation}
where $N$ is skew-symmetric Poisson matrix, $\nabla F$ and $\nabla H$ are
gradients of $F$ and $H$, respectively. A Casimir function $C$ on a Poisson
space is the one that commutes with all the other functions. In order to have
a non-trivial Casimir function, the Poisson matrix $N$ must be degenerate. A
system of ODEs is Hamiltonian if it can be written in the form of Hamilton's
equation%
\begin{equation}
\mathbf{\dot{x}}=\left\{  \mathbf{x},H\right\}  =N\nabla H \label{HamEqn}%
\end{equation}
for $H$ being a real-valued function, called Hamiltonian function,
$\{\bullet,\bullet\}$ being a Poisson bracket and $N$ being the Poisson
matrix. A dynamical system is bi-Hamiltonian if it admits two different
Hamiltonian structures
\begin{equation}
\mathbf{\dot{x}}=N_{1}\nabla H_{2}=N_{2}\nabla H_{1}, \label{biHam}%
\end{equation}
with the requirement that the Poisson matrices $N_{1}$ and $N_{2}$ be
compatible \cite{MaMo84,OLV}.

Space of three dimensional vectors and space of three by three skew-symmetric
matrices are isomorphic. Existence of this isomorphism enables us to identify
a three by three Poisson matrix $N$ with a three dimensional Poisson vector
field $\mathbf{J}$ \cite{EsGhGu16,Gum1}. In this case, the Hamilton's equation
takes the particular form%
\begin{equation}
\mathbf{\dot{x}}=\mathbf{J}\times\nabla H, \label{HamEq3}%
\end{equation}
whereas a bi-Hamiltonian system is in form
\begin{equation}
\mathbf{\dot{x}}=\mathbf{J}_{1}\times\nabla H_{2}=\mathbf{J}_{2}\times\nabla
H_{1}. \label{bi-Ham}%
\end{equation}
and the Jacobi identity turns out to be%
\begin{equation}
\mathbf{J}\cdot(\nabla\times\mathbf{J})=0. \label{jcbv}%
\end{equation}
The following theorem establishes form of a general solution of the Jacobi
identity. For the proof this theorem we refer \cite{AGZ,HB1,HB2,HB3}.

\begin{theorem}
General solution of the Jacobi identity (\ref{jcbv}) is
\begin{equation}
\mathbf{J}=\frac{1}{M}\nabla H_{1} \label{Nsoln}%
\end{equation}
for arbitrary functions $M$ called the Jacobi's last multiplier, and $H_{1}$
called as the Casimir.
\end{theorem}

Existence of the scalar multiple $1/M$ in the solution is a manifestation of
the conformal invariance of Jacobi identity. In the literature, $M$ is called
Jacobi's last multiplier \cite{Go01,Jac1,Jac2,Wh88}. The potential function
$H_{1}$ in Eq.(\ref{Nsoln}) is a Casimir function of the Poisson vector field
$\mathbf{J}$. Any other Casimir of $\mathbf{J}$ has to be linearly dependent
to the potential function$\ H_{1}$ since the kernel is one dimensional.
Substitution of the general solution (\ref{Nsoln}) of $\mathbf{J}$ into the
Hamilton's equations (\ref{HamEq3}) results with
\begin{equation}
\dot{\mathbf{x}}=\frac{1}{M}\nabla H_{1}\times\nabla H_{2}. \label{x1}%
\end{equation}

While writing a non-autonomous system in form of the Hamilton's equations
\eqref{HamEq3}, inevitably, one of the two, Poisson vector or Hamiltonian
function, must depend explicitly on the time variable $t$. The calculation
\[
\frac{d}{dt}H(\mathbf{x},t)=\nabla H(\mathbf{x},t)\cdot\dot{x}+\frac{\partial
}{\partial t}H(\mathbf{x},t)=\nabla H\cdot(\mathbf{J}\times\nabla
H)+\frac{\partial}{\partial t}H(\mathbf{x},t)=\frac{\partial}{\partial
t}H(\mathbf{x},t),
\]
shows that if the time parameter appears only in the Poisson vector, then the
Hamiltonian is a constant of the motion, if the time parameter appears in the
Hamiltonian, then the Hamiltonian fails to be an integral invariant of the system.

\subsection{Nambu-Poisson Systems in $3D$}

In \cite{Nambu}, a ternary operation $\{\bullet,\bullet,\bullet\}$, called
Nambu-Poisson bracket, is defined on the space of smooth functions satisfying
the generalized Leibnitz identity
\begin{equation}
\left\{  F_{1},F_{2},FH\right\}  =\left\{  F_{1},F_{2},F\right\}  H+F\left\{
F_{1},F_{2},,H\right\}  \label{GLI}%
\end{equation}
and the fundamental (or Takhtajan) identity
\begin{equation}
\left\{  F_{1},F_{2},\{H_{1},H_{2},H_{3}\}\right\}  =\sum_{k=1}^{3}%
\{H_{1},...,H_{k-1},\{F_{1},F_{2},H_{k}\},H_{k+1},...,H_{3}\}, \label{FI}%
\end{equation}
for arbitrary functions $F,F_{1},F_{2},H,H_{1},H_{2}$, see \cite{Ta}. A
dynamical system is called Nambu-Hamiltonian with Hamiltonian functions
$H_{1}$ and $H_{2}$ if it can be recasted as%
\begin{equation}
\mathbf{\dot{x}}=\left\{  \mathbf{x},H_{1},H_{2}\right\}  . \label{NHamEqn}%
\end{equation}
By fixing the Hamiltonian functions $H_{1}$ and $H_{2}$, we can write
Nambu-Hamiltonian system (\ref{NHamEqn}) in the bi-Hamiltonian form
\begin{equation}
\mathbf{\dot{x}}=\left\{  \mathbf{x},H_{1}\right\}  ^{H_{2}}=\left\{
\mathbf{x},H_{2}\right\}  ^{H_{1}} \label{NH-2Ham}%
\end{equation}
where the Poisson brackets $\{\bullet,\bullet\}^{H_{2}}$ and $\{\bullet
,\bullet\}^{H_{1}}$ are defined by
\begin{equation}
\left\{  F,H\right\}  ^{{H}_{2}}=\left\{  F,H,H_{2}\right\}  \qquad\left\{
F,H\right\}  ^{{H}_{1}}=\left\{  F,H_{1},H\right\}  , \label{Pois}%
\end{equation}
respectively \cite{Guha06}.

In $3D$, we define a Nambu-Poisson bracket of three functions $F$, $H_{1}$ and
$H_{2}$ as the triple product
\begin{equation}
\left\{  F,H_{1},H_{2}\right\}  =\frac{1}{M}\nabla F\cdot\nabla H_{1}%
\times\nabla H_{2} \label{NambuPois}%
\end{equation}
of their gradient vectors. Note that, the Hamilton's equation (\ref{x1}) is
Nambu-Hamiltonian (\ref{NHamEqn}) with the bracket (\ref{NambuPois}) having
the Hamiltonian functions $H_{1}$ and $H_{2}$ \cite{Guha06, TeVe04}. If the
function $F$ in (\ref{NambuPois}) is taken as the coordinate functions, then
it becomes the Lie-Poisson bracket on $\mathbb{R}^{3}$ of two functions
$H_{1}$ and $H_{2}$ identified with $\left(  \mathbb{R}^{3}\right)  ^{\ast}$
using the dot product \cite{BlMoRa13}, that is
\[
\left\{  H_{1},H_{2}\right\}  _{LP}=\frac{1}{M}\mathbf{x}\cdot\nabla
H_{1}\times\nabla H_{2}.
\]
The following theorem establishes the link between the existence of the
Hamiltonian structure of a dynamical system and the existence of the Jacobi's
last multiplier. For the proof of the assertion we cite \cite{EsGhGu16,Gao}.

\begin{theorem}
\label{3} A three dimensional dynamical system $\dot{\mathbf{x}}=\mathbf{X}$
having a time independent first integral is Hamiltonian, bi-Hamiltonian hence
Nambu Hamiltonian if and only if there exist a Jacobi's last multiplier $M$
which makes $M\mathbf{X}$ divergence free.
\end{theorem}

\subsection{Metriplectic Systems in $3D$}

Let $G$ be a positive semi definite symmetric matrix on an Euclidean space,
and consider the symmetric bracket of two functions
\[
\left(  F,S\right)  =\nabla F\cdot G\nabla S.
\]
In terms of this symmetric bracket, we define a metric or a gradient system by%
\begin{equation}
\mathbf{\dot{x}}=\left(  \mathbf{x},S\right)  =G\nabla S.
\end{equation}
The generating function, usually called the entropy, is not a conserved
quantity for the system instead we have $\dot{S}=\left(  S,S\right)  \geq0$,
see \cite{Fi05}.

The representation of a dynamical system as a metriplectic system requires two
geometrical structures namely a Poisson structure $N$ and a metric structure
$G$. Metriplectic bracket is the sum of the two brackets
\[
\{\{F,E\}\}=\{F,E\}+\lambda\left(  F,E\right)  =\nabla F\cdot N\nabla
E+\lambda\nabla F\cdot G\nabla E,
\]
for any scalar $\lambda$. There are extensive studies on the metriplectic
systems see, for example, \cite{BiBoPuTu07,BlMoRa13,Br88,Fi05,Gu07,Mo84,Mo86,Ka84}. The
metriplectic structures also called with the name GENERIC
\cite{GrOt97}. The metriplectic structure satisfies the Leibnitz
identity for each entry hence it is an example of a Leibnitz bracket
\cite{OrPl04}. We refer \cite{Mo09} for a brief history of metriplectic
structures and more.

There are two types of metriplectic systems in the literature. One of them is
the one governed by so called a generalized free energy $F$ which is the
difference of a Hamiltonian function $H$ and a entropy function $S$ . In this
case, we require that $\nabla S$ lives in the kernel of $N$ and $\nabla H$
lives in the kernel of $G$, that is
\begin{equation}
N\nabla S=0,\text{ \ \ }G\nabla H=0. \label{mc1}%
\end{equation}
The equation of motion is given by
\[
\mathbf{\dot{x}}=\{\{\mathbf{x},F\}\}=\{\mathbf{x},F\}+\left(  \mathbf{x}%
,F\right)  =\{\mathbf{x},H\}-\left(  \mathbf{x},S\right)  .
\]
Note that, for the dynamics governed by the metriplectic bracket, we have the
conservation law $\dot{H}=\{\{H,F\}\}=0$ and the dissipation $\dot
{S}=\{\{S,F\}\}\leq0$ . We note that a weaker version of the condition
(\ref{mc1}) can be given by
\begin{equation}
N\nabla S+G\nabla H=0\text{.} \label{mc2}%
\end{equation}
The second type of the metriplectic systems is generated by a single function,
say $H$, and written as
\begin{equation}
\mathbf{\dot{x}}=\{\{\mathbf{x},H\}\}=\{\mathbf{x},H\}+\lambda\left(
\mathbf{x},H\right)  \label{mp2nd}%
\end{equation}
without any restriction on $H$ as given in (\ref{mc1}) or (\ref{mc2}).

If the Hamiltonian (reversible) part of the dynamics can be written in the
terms of Nambu-Poisson bracket we may rewrite the system as
\begin{equation}
\label{metri1}\mathbf{\dot{x}}=\{\mathbf{x},H_{1},H_{2}\}+\left(
\mathbf{x},S\right)  =\frac{1}{M}\nabla H_{1}\times\nabla H_{2}-G\nabla S,
\end{equation}
where $M$ is the Jacobi's last multiplier \cite{Bi08}. In this case, one may take $S$
equals to $H_{1}$ or $H_{2}$.

\section{Examples}

\subsection{Reduced three-wave interaction problem}

The reduced three-wave interaction model \cite{Go01,PR} is given by the system
of ODEs
\begin{equation}
\label{Rabi1}%
\begin{cases}
\dot{x}=-2y^{2}+\gamma x+z+\delta y\\
\dot{y}=2xy+\gamma y-\delta x\\
\dot{z}=-2xz-2z.
\end{cases}
\end{equation}
where three quasisynchronous waves interact in a plasma with quadratic
nonlinearities. In \cite{BR}, this model is studied by means of Painlev\'{e}
method. In \cite{GRZ}, the existence of first integrals for this and other
systems were investigated by proposing an ansatz for the first integral which
explicitly involves a pre-set dependence on a particular phase space
coordinate. We show how their results can be obtained in a more simplified
manner using Darboux polynomials. We, additionally, present bi-Hamiltonian and
metriplectic realizations of the model.

\begin{proposition}
The three dimensional reduced three-wave interaction problem (\ref{Rabi1}) has
the following first integrals.

\begin{enumerate}
\item If $\delta=\hbox{ arbitrary}$, $\gamma=0$, then $I=e^{2t}z\big(y-\delta
/2\big).$

\item If $\delta=\hbox{ arbitrary}$, $\gamma=-1$, then $I=e^{2t}(x^{2}%
+y^{2}+z).$

\item If $\delta=\hbox{ arbitrary}$, $\gamma=-2$, then $I=e^{4t}(x^{2}%
+y^{2}+2/\delta\,yz).$

\item If $\delta=0$, $\gamma=\hbox{ arbitrary}$, then $I=e^{2-\gamma}yz.$

\item If $\delta=0$, $\gamma=-1$, then $I_{1}=e^{2t}(x^{2}+y^{2}+z),$
$I_{2}=e^{3t}yz.$
\end{enumerate}
\end{proposition}

In order to prove this assertion, we recall the eigenvalue problem (\ref{dp})
associated with the system (\ref{Rabi1}) where $g$ is a second degree
polynomial of the form
\begin{equation}
g=Ax^{2}+By^{2}+Cz^{2}+Exy+Fxz+Gyz+Jx+Ky+Lz.
\end{equation}
Equating coefficients then leads to the following set of equations
\begin{align}
&  A=B,\qquad E=F=C=0\label{c3}\\
&
\begin{cases}
2A\gamma-E\delta=\lambda A,\qquad2B\gamma+E\delta-2J=\lambda B,\\
F-4C=\lambda C,\qquad2A\delta+2E\gamma-2B\delta+2K=\lambda E,\\
2A+(\gamma-2)F-G\delta-2L=\lambda F,\qquad E+F\delta+(\gamma-2)G=\lambda G
\end{cases}
\label{c2}\\
&  J\gamma-K\delta=\lambda J,\qquad J\delta+K\gamma=\lambda K,\qquad
J-2L=\lambda L. \label{c1}%
\end{align}
for the third order, the second order, and the linear terms, respectively. We
distinguish a number of cases following from the solutions of the system
(\ref{c3})-(\ref{c1}) for specific parameter values. These cases will
determine the integrals of the reduced systems by following the theorem
(\ref{1}). In the first three cases, $\delta$ is arbitrary, and we study on
three different values of $\gamma$, namely $0,-1$ and $-2$. For the remaining
cases wherein $\delta=0$ one can identify explicitly Darboux functions of the
associated vector field, with associated eigenpolynomials which are not of
degree zero.

\subsubsection{Case 1: $\delta$ is arbitrary and $\gamma=0$}

The choices of $\delta$ is arbitrary and $\gamma=0$ reduce the system of
equations (\ref{c3})-(\ref{c1}) to the following list
\[
A=B=C=E=F=K=J=0,\qquad L=-{\frac{\delta}{2}}G,\qquad\lambda=-2
\]
where $G$ is arbitrary function. Additionally, by choosing $G=1$, we obtain
the eigenfunction
\[
g=zy-{\frac{\delta}{2}z}.
\]
The condition (\ref{clue}) translates to the following requirement
$-r+n\lambda=0$. For $r=-1$, we have $n=1/2$, so that an integral of the
motion equals to $e^{t}(zy-{\frac{\delta}{2}z})^{\frac{1}{2}}$. As any
function of this integrating factor is also a first integral we write the
integral as%
\begin{equation}
I=e^{2t}\left(  zy-{\frac{\delta}{2}z}\right)  . \label{G1}%
\end{equation}

We change the dependent variable $z$ by $w$ according to $w=e^{2t}z$. In this
case, the system (\ref{Rabi1}) turns out to be a non-autonomous system
\begin{equation}
\label{case1Ham}%
\begin{cases}
\dot{x}=-2y^{2}+we^{-2t}+\delta y\\
\dot{y}=2xy-\delta x\\
\dot{w}=-2xw
\end{cases}
\end{equation}
whereas the integral $I$ in (\ref{G1}) becomes time independent Hamiltonian of
the system given by
\[
H_{1}=wy-{\frac{\delta}{2}}w.
\]
This reduced system is divergence free, hence, according to the theorem
(\ref{3}), it is bi-Hamiltonian (\ref{biHam}) and Nambu-Poisson (\ref{NHamEqn}%
). To exhibit these realizations, we need to introduce a second time dependent
Hamiltonian function
\[
H_{2}=x^{2}+y^{2}+e^{-2t}w
\]
of the system (\ref{case1Ham}). Note that, the system is divergce free, hence
Jacobi's last multiplier for the system is a constant function, say $M=1$. So
that, the system is in the form of cross product of two gradients%
\[
\left(  \dot{x},\dot{y},\dot{w}\right)  ^{T}=\nabla H_{1}\times\nabla
H_{2}=\mathbf{J}_{1}\times\nabla H_{2}=\mathbf{J}_{2}\times\nabla H_{1}%
\]
in the form of bi-Hamiltonian and Nambu-Poisson forms (\ref{x1}) with Poisson
vector fields $\mathbf{J}_{1}=\nabla H_{1}$ and $\mathbf{J}_{2}=-\nabla H_{2}%
$, respectively. Since the first Hamiltonian is autonomous, the second one has
to, evidently, be time dependent. Note that, this second time dependent
Hamiltonian $H_{2}$ can not be observed as a consequence of the theorem
(\ref{1}), because it is not an integral invariant of the system.

At this point, we make a break to the cases and discuss the metriplectic
structure of the system (\ref{Rabi1}) starting and inspiring from the
bi-Hamiltonian/Nambu formulation of its particular case (\ref{case1Ham}). The
proof of the following assertion is a matter of direct calculation

\begin{proposition}
\label{rtwmp} The reduced three-wave interaction problem (\ref{Rabi1})is in
bi-Hamiltonian/Nambu metriplectic formulation (\ref{metri1}) given by
\begin{equation}
\left(  \dot{x},\dot{y},\dot{z}\right)  ^{T}=\nabla H_{1}\times\nabla
H_{2}-G\nabla H_{2}. \label{Metriplectic1}%
\end{equation}
where the Hamiltonian functions are $H_{1}=zy-{\frac{\delta}{2}}z$, and
$H_{2}=x^{2}+y^{2}+e^{-2t}z$, and the metric tensor is
\[
G=%
\begin{pmatrix}
-\gamma/2 & 0 & 0\\
0 & -\gamma/2 & 0\\
0 & 0 & 2ze^{2t}%
\end{pmatrix}
.
\]

\end{proposition}

In (\ref{Metriplectic1}), the metriplectic structure in of the second kind.
Note that, by replacing the roles of $H_{1}$ and $H_{2}$ in
(\ref{Metriplectic1}), up to some modifications in the definition of the
metric, we may also generate the system (\ref{Rabi1}) by the Hamiltonian
$H_{1}$ as well. This case will be presented in the case $5$.

\subsubsection{Case 2: $\delta$ is arbitrary and $\gamma=-1$}

In the case $\delta$ is arbitrary and $\gamma=-1$, the system of equations
(\ref{c3})-(\ref{c1}) becomes
\[
C=E=F=G=K=J=0,\qquad A=B=L,\qquad\lambda=-2
\]
so that the eigenfunction becomes $A(x^{2}+y^{2}+z)$. Hence, the condition for
$I$ to be a first integral, namely $-r+n\lambda=0$ implies $n={\frac{1}{2}}$
and $r=-1$. The corresponding first integral is then given by
\begin{equation}
I=e^{2t}(x^{2}+y^{2}+z). \label{G2}%
\end{equation}

We make the change of dependent variables
\[
u=xe^{t},\text{ \ \ }v=ye^{t},\text{ \ \ }w=ze^{2t}%
\]
and rescale the time variable by $\bar{t}=e^{t}$, then we arrive the non-autonomous system
\begin{equation}
\left\{
\begin{array}
[c]{c}%
\acute{u}=-2v^{2}+w+\delta v\bar{t}\\
\acute{v}=2uv-\delta u\bar{t}\\
\acute{w}=-2uw
\end{array}
\right.  \label{Rtwip2b}%
\end{equation}
where prime denotes the derivative with respect to the new time variable
$\bar{t}=e^{t}$. In this coordinates, the integral (\ref{G2}) is autonomous
\[
H_{1}=u^{2}+v^{2}+w.
\]
Note that, the system (\ref{Rtwip2b}) is divergence free, hence we can take
the Jacobi's last multiplier $M$ as the unity. Hence, we argue that, there
exist a second Hamiltonian which enables us to write the system (\ref{Rtwip2b}%
) in bi-Hamiltonian/Nambu formulation. After a straight forward calculation,
we arrive a non-autonomous Hamiltonian
\[
H_{2}=vw+\delta\frac{v^{2}}{2}\bar{t}-\delta\frac{u^{2}}{2}\bar{t}%
\]
which enables us to write the system (\ref{Rtwip2b}) as a bi-Hamiltonian
(\ref{biHam}) and Nambu-Poisson (\ref{NHamEqn}) system
\[
\left(  \acute{u},\acute{v},\acute{w}\right)  ^{T}=\nabla H_{1}\times\nabla
H_{2}=\mathbf{J}_{1}\times\nabla H_{2}=\mathbf{J}_{2}\times\nabla H_{1}%
\]
where the Poisson vectors are $\mathbf{J}_{1}=\nabla H_{1}$ and $\mathbf{J}%
_{2}=-\nabla H_{2}$, respectively.

\subsubsection{Case 3: $\delta$ is arbitrary and $\gamma=-2$}

For the above choice of parameters $\delta$ is arbitrary and $\gamma=-2$, it
may be verified that, the system of equations (\ref{c3})-(\ref{c1}) turn out
to be
\[
C=E=F=J=K=L=0,\qquad A=B,\qquad G={\frac{2}{\delta}}A,\qquad\lambda=-4.
\]
This leads to the eigenfunction $A(x^{2}+y^{2}+{\frac{2}{\delta}y}z)$, so that
choosing $A=1$ we get the following first integral
\begin{equation}
I=e^{4t}(x^{2}+y^{2}+{\frac{2}{\delta}y}z). \label{G3}%
\end{equation}
To arrive the Hamiltonian form of this system, we first make the substitutions
$u=xe^{2t},$ $v=ye^{2t},$ $w=ze^{2t}$ which results with the non-autonomous
divergence free system%
\begin{equation}
\left\{
\begin{array}
[c]{c}%
\dot{u}=-2v^{2}e^{-2t}+w+\delta v\\
\dot{v}=2uve^{-2t}-\delta u\\
\dot{z}=-2uwe^{-2t}%
\end{array}
\right.  . \label{case3}%
\end{equation}
Actually, the system (\ref{case3})\ is a bi-Hamiltonian (\ref{biHam}) and
Nambu-Poisson (\ref{NHamEqn}) system with the introductions of Hamiltonian
functions
\[
H_{1}=\frac{\delta}{2}\left(  u^{2}e^{-2t}+v^{2}e^{-2t}+w\right)  \text{,
\ \ }H_{2}=u^{2}+v^{2}+{\frac{2}{\delta}vw,}%
\]
where the second Hamiltonian is the integral (\ref{G3}).

\subsubsection{Case 4: $\delta=0$ and $\gamma$ is arbitrary}

It is a straightforward matter to verify that the following functions
$g_{\alpha}\;(\alpha=1,2)$ are Darboux polynomials whose associated
eigenpolynomials $\lambda_{\alpha}$'s are
\begin{equation}
\label{gs}g_{1}=y,\;\lambda_{1}=2x-1\text{, \ \ and \ \ }g_{2}=z,\;\lambda
_{2}=-2x-2,
\end{equation}
if $\delta=0$ and $\gamma$ is arbitrary. The condition (\ref{clue}) now leads
to
\[
0=-r+\sum_{\alpha}n_{\alpha}g_{\alpha}\Rightarrow-r+n_{1}(2x+\gamma
)+n_{2}(-2x-2)=0.
\]
Setting $r=-1$ we obtain the following equations:
\[
n_{1}-n_{2}=0,\;\;\gamma n_{1}-2n_{2}+1=0
\]
leading to $n_{1}=n_{2}={\frac{1}{2-\gamma}}$. The corresponding first
integral is
\begin{equation}
I=e^{(2-\gamma)t}yz. \label{G4}%
\end{equation}
In order to exhibit the Hamiltonian formulation of the system, we define
$u=xe^{-\gamma t},$ \ \ $v=ye^{-\gamma t},$ \ \ $w=ze^{2t}$ then we have a non-autonomous divergence free system
\begin{equation}%
\begin{cases}
\dot{u}=-2v^{2}e^{\gamma t}+we^{-(2+\gamma)t}\\
\dot{v}=2uve^{\gamma t}\\
\dot{z}=-2uwe^{\gamma t}%
\end{cases}
\end{equation}
with the Hamiltonian $H_{2}=vw$. The bi-Hamiltonian (\ref{biHam}) and
Nambu-Poisson (\ref{NHamEqn}) structure of the system can be realized after
the introduction of the second (time dependent) Hamiltonian
\[
H_{1}=u^{2}e^{\gamma t}+v^{2}e^{\gamma t}+e^{-\left(  2+\gamma\right)  t}w.
\]

\subsubsection{Case 5: $\delta=0$ and $\gamma=-1$}

For this case, in addition to $g_{1},g_{2}$ given in (\ref{gs}), we have
another Darboux polynomial
\[
g_{3}=x^{2}+y^{2}+z,\;\;\lambda_{3}=-2.
\]
The condition (\ref{clue}) becomes
\[
2(n_{1}-n_{2})-(n_{1}+2n_{2}+2n_{3})=r.
\]
We make the standardization $r=-1$ and obtain the following set of equations
\[
n_{1}=n_{2}\;\mbox{and}\;n_{1}+2n_{2}+2n_{3}=1
\]
or, in other words, $3n_{1}+2n_{3}=1$ which leads to the following subcases:
(a) $n_{3}=0\;$and$\;n_{1}=n_{2}={\frac{1}{3}}$, and (b) $n_{1}=n_{2}%
=0\;$and$\;n_{3}={\frac{1}{2}}$. So that, we have two time dependent integrals
of the motion
\begin{equation}
I_{1}(x,y,z)=e^{t}(yz)^{{\frac{1}{3}}}\text{ and \ }I_{2}(x,y,z)=e^{t}%
(x^{2}+y^{2}+z)^{{\frac{1}{2}}}. \label{Int12}%
\end{equation}
We make the change of variables $u=xe^{t}$, $v=ye^{t}$, and $w=ze^{2t}$ and
rescale the time variable by $\bar{t}=e^{t}$, then arrive the autonomous
system
\begin{equation}
\label{systemp}%
\begin{cases}
\acute{u}=-2v^{2}+w\\
\acute{v}=2uv\\
\acute{w}=-2uw
\end{cases}
\end{equation}
where prime denotes the derivative with respect to the new time variable
$\bar{t}=e^{t}$. Note that, this system is divergence free, hence we can take
the Jacobi's last multiplier as the unity. In the new coordinate system, the
integrals of the system\ (\ref{Int12}) become the Hamiltonian functions of the
system given by
\begin{equation}
H_{1}=vw\text{, \ \ }H_{2}=u^{2}+v^{2}+w. \label{Ham5}%
\end{equation}
This enables us to write the system (\ref{systemp}) in bi-Hamiltonian
(\ref{biHam}) and Nambu-Poisson (\ref{NHamEqn}) form.

Note that, as a particular case of the proposition (\ref{rtwmp}), we show how
the reduced three-wave interaction model (\ref{Rabi1}) with $\delta
=0$\textbf{\ }and $\gamma=-1$ given by
\begin{equation}%
\begin{pmatrix}
\dot{x}\\
\dot{y}\\
\dot{z}%
\end{pmatrix}
=%
\begin{pmatrix}
-2y^{2}+z\\
2xy\\
-2xz
\end{pmatrix}
+%
\begin{pmatrix}
-x\\
-y\\
-2z
\end{pmatrix}
\label{rtwipmetri}%
\end{equation}
can be put in a metriplectic realization of the second kind (\ref{mp2nd}).
Note that, the first term at the right hand side is the conservative part of
the system with two Hamiltonian functions $H_{1}=yz$ and $H_{2}=x^{2}+y^{2}+z
$ inspired from the ones in (\ref{Ham5}). This enables us to write the system
(\ref{rtwipmetri}) in two different ways. In the first one, we take $H_{2}$ as
the Casimir function of the system and $H_{1}$ as the Hamiltonian system with
Poisson vector $\mathbf{J}_{2}=-\nabla H_{2}$. Hence, the second term on the
right hand side can be described by a dissipative term by taking the metric
two-form as
\[
G=%
\begin{pmatrix}
0 & \frac{x}{z} & 0\\
\frac{x}{z} & 0 & 1\\
0 & 1 & \frac{z}{y}%
\end{pmatrix}
\]
where $\lambda=-1$. In this case the reduced three wave interaction model
(\ref{rtwipmetri}) can be written as
\[
\mathbf{\dot{x}}=\mathbf{J}_{2}\times\mathbf{\nabla}H_{1}-G\mathbf{\nabla
}H_{1}.
\]

\subsection{Rabinovich system}

This is described by the following system of equations:
\begin{equation}
\label{Rabi2}%
\begin{cases}
\dot{x}=hy-\nu_{1}x+yz\\
\dot{y}=hx-\nu_{2}y-xz\\
\dot{z}=-\nu_{3}z+xy,
\end{cases}
\end{equation}
where $h$ and $\nu_{i}$ are real constants. We shall very briefly illustrate
how the results of \cite{GRZ} for this system may be derived by the Darboux
integrability method. In addition, we will show that the Rabinovich system
\eqref{Rabi2} can be written as a bi-Hamiltonian/Nambu metriplectic form.

Consider the vector field $X$ generating the Rabinovich system (\ref{Rabi2}).
We note that application of $X$ to the function $g_{1}=y^{2}+z^{2}$ yields
\begin{equation}
X\left(  g_{1}\right)  =2hxy-2(\nu_{2}y^{2}+\nu_{3}z^{2}).
\end{equation}
Consequently $g_{1}$ becomes a Darboux polynomial when $h=0,\nu_{2}=\nu_{3} $.
In this case, the eigenpolynomial being of degree zero \textit{viz}
$\lambda=-2\nu_{3}$. We are lead to the first integral%
\begin{equation}
\label{G6}I_{1}=e^{2\nu_{3}t}(y^{2}+z^{2})
\end{equation}
of the system (\ref{Rabi2}) when\ \ $h=0$,\ $\nu_{2}=\nu_{3}$\ with\ $\nu_{1}
$ and $\nu_{3}$\ being arbitrary. The application of the vector field $X$
generating the Rabinovich system (\ref{Rabi2}) on the polynomial $g_{2}%
=x^{2}+y^{2}$ results with
\[
X\left(  g_{2}\right)  =4hxy-2(\nu_{1}x^{2}+\nu_{2}y^{2}).
\]
Consequently, $g_{2}$ becomes a Darboux polynomial when $h=0,\nu_{1}=\nu_{2}
$. In this case, the eigenpolynomial being of degree zero \textit{viz}
$\lambda=-2\nu_{1}$. We are lead to the first integral%
\begin{equation}
\label{G7}I_{2}=e^{2\nu_{1}t}(x^{2}+y^{2})
\end{equation}
of the system (\ref{Rabi2}) when\ \ $h=0$,\ $\nu_{1}=\nu_{2}$\ with\ $\nu_{1}
$,$\nu_{3}$\ being arbitrary.

Let us transform the Rabinovich system (\ref{Rabi2}) in a form where we can
write it as a bi-Hamiltonian/Nambu system. For the case of \ $\nu_{1}=\nu
_{2}=v_{3}=v$, we have two integrals $I_{1}$ and $I_{2}$ of the system
(\ref{Rabi2}). In this case, we apply a coordinate change
\[
u=xe^{vt},v=ye^{vt},\text{ \ \ }w=ze^{vt}%
\]
with the time rescaling $\bar{t}=\frac{1}{v}e^{vt}$ with $v\neq0$, then the
system turns out to be a divergence free system
\begin{equation}
\acute{u}=vw,\text{ \ \ }\acute{v}=-uw,\text{ \ \ }\acute{w}=uv.
\label{DivfreeRabi}%
\end{equation}
In this case the integrals of motion given in (\ref{G6}) and (\ref{G7}) become
the Hamiltonian functions of the system, namely
\[
H_{1}=\frac{1}{2}(v^{2}+w^{2}),\text{ \ \ }H_{2}=\frac{1}{2}(u^{2}+v^{2}).
\]
Hence we can write (\ref{DivfreeRabi}) as in the form of bi-Hamiltonian
(\ref{biHam}) and Nambu-Poisson (\ref{NHamEqn}) form
\begin{equation}
\left(  \acute{u},\acute{v},\acute{w}\right)  ^{T}=\nabla H_{1}\times\nabla
H_{2} \label{Rabinovich2}%
\end{equation}
with Jacobi' last multiplier being the unity, see also \cite{CrPu07}. For
another discussion on the case where $h$ is nonzero and $\nu_{1}=\nu_{2}%
=v_{3}=0$, we refer \cite{Tu12}.

In the following proposition, inspiring from the bi-Hamiltonian/Nambu form
(\ref{Rabinovich2}) of the transformed system (\ref{DivfreeRabi}), we are,
now, exhibiting a metriplectic realization of the Rabinovich system
(\ref{Rabi2}).

\begin{proposition}
\label{rabi} The Rabinovich system (\ref{Rabi2}) is in bi-Hamiltonian/Nambu
metriplectic formulation (\ref{metri1}) given by
\begin{equation}
\left(  \dot{x},\dot{y},\dot{z}\right)  ^{T}=\nabla H_{1}\times\nabla
H_{2}-G\nabla H_{1}. \label{metrirabi}%
\end{equation}
where the Hamiltonian functions are $H_{1}=\frac{1}{2}(x^{2}+y^{2})$, and
$H_{2}=\frac{1}{2}(y^{2}+z^{2})$, and the metric tensor is
\[
G=%
\begin{pmatrix}
\nu_{1} & -h & 0\\
-h & \nu_{2} & \frac{z\nu_{3}}{y}\\
0 & \frac{z\nu_{3}}{y} & 0
\end{pmatrix}
.
\]

\end{proposition}

The metriplectic formulation (\ref{metrirabi}) of the Rabinovich system
(\ref{Rabi2}) is of the second kind. As in the case of the reduced three-wave
interaction problem, one may generate (\ref{Rabi2}) by the Hamiltonian $H_{2}$
instead of $H_{1}$ by adopting a new metric.

\subsection{Hindmarsh-Rose model}

The Hindmarsh-Rose model of the action potential which is a modification of
Fitzhugh model was proposed as a mathematical representation of the bursting
behaviour of neurones, and was expected to simulate the repetitive, patterned
and irregular activity seen in molluscan neurones \cite{HiRo84}. The
Hindmarsh-Rose model consists of a system of three autonomous differential
equations, with mild nonlinearities for modelling neurons that exhibit
triggered firing. The usual form of the equations are
\begin{equation}%
\begin{cases}
\dot{x}=y+\phi(x)-z-C\\
\dot{y}=\psi(x)-y\\
\dot{z}=r(s(x-x_{R})-z)
\end{cases}
\label{HR1}%
\end{equation}
where $\phi(x)=ax^{2}-x^{3}$ and $\psi(x)=1-bx^{2}$. Here $C$ is a control
parameter, while of the remaining five parameters $s$ and $x_{R}$ are usually
fixed. We re-write them in the following form appending two extra parameters%
\begin{equation}%
\begin{cases}
\dot{x}=y-z-ax^{3}+bx^{2}+\alpha\\
\dot{y}=\beta-dx^{2}-y\\
\dot{z}=px-rz-\gamma
\end{cases}
\label{HR2}%
\end{equation}
Here $\alpha,\beta,\gamma,a,b,d,p,r$ are parameters. Unfortunately we have not
found a first integral with $a\neq0$, which is the dominant nonlinear term here.

\begin{proposition}
The reduced Hindmarsh-Rose system
\begin{equation}%
\begin{cases}
\dot{x}=y-z+bx^{2}+\alpha\\
\dot{y}=\beta-dx^{2}-y\\
\dot{z}=px-rz-\gamma
\end{cases}
\label{HR3}%
\end{equation}
has the following first integrals.

\begin{enumerate}
\item If $p=0$, then the first integral of the system (\ref{HR3}) is
$I=e^{rt}(rz+\gamma).$

\item If $d=0$ then $I=e^{t}(y-\beta).$

\item If $d,\beta,\gamma$ are arbitrary, $b=-d,p=-2,\alpha=\beta+\gamma$ and
$r=1$, then $I=e^{2t}(x-y+z).$

\item If $\alpha,\gamma,p$ and $b$ are arbitrary, and when
$d=2b,r=-(p+1),\beta=2(\frac{\gamma}{p}-\alpha)$ then $I=e^{-t}(2x+y+\frac
{2z}{p}).$

\item If $\beta,\gamma,r,b,d$ are arbitrary, and%
\[
\alpha=-\frac{b(\gamma d+\beta d-b\beta+r\beta b)}{d(d-b+br)}\qquad
\text{and}\qquad p=\frac{(b-d)(d-b+br)}{b^{2}}%
\]
then the first integral becomes
\[
I=e^{\frac{2(b-d)}{b}}(Ax^{2}+By^{2}+Cz^{2}+Exy+Fxz+Gyz)
\]
where the coefficients of the polynomial are given by
\begin{align*}
A  &  =-\frac{(b-d)(d-b+br)}{b(-d+2b+br)},\qquad B=-\frac{b(b-d)(d-b+br)}%
{d^{2}(-d+2b+br)},\\
C  &  =-\frac{b(b-d)}{(d-b+br)(-d+2b+br)},\qquad E=-2\frac{(b-d)(d-b+br)}%
{d(-d+2b+br)},\\
F  &  =2\frac{b-d}{-d+2b+br},\qquad G=2\frac{b(b-d)}{d(-d+2b+br)}.
\end{align*}

\item If $p=0$ , $b=d$, and $\beta,\gamma,r$ are arbitrary and $\alpha
=-\frac{\beta r+\gamma}{r}$ then $I=rx+ry-z.$ When, additionally, $r=-1$, then
$I=x+y+z.$
\end{enumerate}
\end{proposition}

To prove these assertions, one may take the total time derivatives of the
integrals and show that they are zero. Starting with the integrals presented in the
previous proposition, we are achieving to write
the Hindmarsh-Rose model (\ref{HR2}) in a metriplectic form of the second kind in the following proposition.

\begin{proposition}
\label{metriHR} The Hindmarsh-Rose model (\ref{HR2})\ (with $r=-1$ and
$\alpha=\beta-\gamma$) is in bi-Hamiltonian/Nambu metriplectic formulation
(\ref{metri1}) given by
\begin{equation}
\left(  \dot{x},\dot{y},\dot{z}\right)  ^{T}=\nabla H_{1}\times\nabla
H_{2}-G\nabla H_{1}.
\end{equation}
where the Hamiltonian functions are $H_{1}=x+y+z$, and $H_{2}=yz-\gamma
y-\beta z$, and the metric tensor is
\[
G=%
\begin{pmatrix}
ax^{3}-bx^{2} & 0 & 0\\
0 & dx^{2} & 0\\
0 & 0 & -px
\end{pmatrix}
.
\]

\end{proposition}

\subsection{Oregonator model}

The Oregonator model was developed by Field and Noyes \cite{FN} to illustrate
the mechanism of the Belousov-Zhabotinsky oscillatory reaction. The model can
be expressed in terms of coupled three ordinary differential equations
\begin{equation}
\left\{
\begin{array}
[c]{c}%
\dot{x}={\frac{1}{\epsilon}}(x+y-qx^{2}-xy)\\
\dot{y}=-y+2hz-xy\\
\dot{z}={\frac{1}{p}}(x-z).
\end{array}
\right.  \label{Ore1}%
\end{equation}
that describe the complex dynamics of the reaction process. In the physical
model considered, all the parameters $\epsilon,q,p,h$ are positive. However, from a
purely mathematical point of view, allowing the parameters to be negative, we
have obtained a first integral
\begin{equation}
I=e^{2t}(x+y+z),
\end{equation}
for the parameters$\;q=0,\epsilon=p=-1$ and$\;h=-{\frac{3}{2}}$ as may be
easily verified.

We will write Oregonator model in the Hamiltonian formulation as follows. At
first, we change the coordinates according to
\[
u=xe^{2t},\text{ \ \ }v=ye^{2t},\text{ \ \ }w=e^{2t}z.
\]
which enables us to write the system (\ref{Ore1}) as the following
nonautonomous form
\begin{align}
\dot{u}  &  =u-v+uve^{-2t}\nonumber\\
\dot{v}  &  =v-3w-uve^{-2t}\nonumber\\
\dot{w}  &  =3w-u
\end{align}
with a time independent first integral $H=u+v+w$. Then we introduce the non-autonomous Poisson matrix%
\[
N=%
\begin{pmatrix}
0 & uve^{-2t}-v & u\\
v-uve^{-2t} & 0 & -3w\\
-u & 3w & 0
\end{pmatrix}
\]
then the system (\ref{Ore1}) is in form Hamilton's equation (\ref{HamEqn})
given by $\mathbf{\dot{u}}=P\nabla H$.

\section{Conclusions}

In this paper, we have reviewed some technical details of the integrability
and Hamiltonian representations of the $3D$ systems. Then, we have applied
these theoretical results, especially the Darboux polynomials, to derive the
first integrals of $3D$ polynomial systems the reduced three-wave interaction
problem, Rabinovich system, Hindmarsh-Rose model and Oregonator model. Then we have
achieved to exhibit Hamiltonian, and metriplectic realizations of the systems.

\end{document}